\documentclass[pdflatex, sn-nature]{sn-jnl}

\usepackage{graphicx}%
\usepackage{multirow}%
\usepackage{amsmath,amssymb,amsfonts}%
\usepackage{amsthm}%
\usepackage{mathrsfs}%
\usepackage[title]{appendix}%
\usepackage{xcolor}%
\usepackage{textcomp}%
\usepackage{manyfoot}%
\usepackage{booktabs}%

\newtheorem{theorem}{Theorem}[section]
\newtheorem{lemma}[theorem]{Lemma}

\providecommand{\E}{\mathbb{E}}
\providecommand{\R}{\mathbb{R}}
\providecommand{\Var}{\mathrm{Var}}
\providecommand{\Bias}{\mathrm{Bias}}
\providecommand{\MSE}{\mathrm{MSE}}
\providecommand{\MISE}{\mathrm{MISE}}
\providecommand{\dd}{\mathrm{d}}
\providecommand{\ee}{\mathrm{e}}

\begin{document}

\title[Universal Prediction of Extremes]{The data-driven extreme value distribution: non-parametric tail estimation with a derived stability criterion}

\author*[1]{\fnm{Michael} \sur{Sandbichler}}\email{m.sandbichler@datalabhell.ac.at}
\author[1]{\fnm{Tobias} \sur{Hell}}\email{t.hell@datalabhell.ac.at}

\affil[1]{\orgname{Data Lab Hell GmbH}, \orgaddress{\city{Zirl}, \country{Austria}}}

\abstract{
Quantifying the likelihood of extreme events underpins risk assessment, yet classical Extreme Value Theory relies on asymptotic assumptions that fail in the data-sparse, non-stationary regimes practitioners increasingly face.
We introduce the Data-Driven Extreme Value Distribution (DDEVD), a non-parametric estimator that aggregates all observations metastatistically and reconstructs the base distribution with a kernel, removing parametric tail assumptions.
We derive its optimal bandwidth and prove a stability law $m < C\,n^{1+\gamma/2}$ relating reliable extrapolation to the extreme value index $\gamma$.
In sub-hourly Alpine precipitation, DDEVD recovers stable 100-year return levels from single decades (calibration ratio $0.96$), departing from the full-record reference by over $50\,\%$ in fewer than one window in fifty --- versus one in five for a GEV fit.
In metallurgical micrographs, it matches a generalised extreme-value fit on the safety-relevant grain-size tail, where the standard log-normal over-predicts by $58\,\%$ at $1\,\mathrm{cm}^{2}$.
}

\keywords{Extreme Value Theory, Non-parametric statistics, Tail risk estimation, Metastatistics}

\maketitle

Extreme events set the design limits of our infrastructure, the resilience of our ecosystems, the safety of our materials and the stability of our economies.
Whether sizing a flood defence in the Alpine region or qualifying a refractory-metal component for a medical device, the critical variable is the event that has not yet happened: the ``300-year'' flood or the maximal microscopic structural anomaly that initiates failure.
The statistical analysis of such extremes is a cornerstone of risk assessment~\cite{castillo2012extreme, embrechts2013modelling}, yet it faces a fundamental paradox: the most dangerous events are also the rarest, and therefore the hardest to predict from historical data.
Classical Extreme Value Theory (EVT) addresses this through asymptotic arguments --- the Block Maxima Method (BMM)~\cite{gumbel1958statistics} and the Peak Over Threshold (POT) method~\cite{pickands1975statistical} both assume long, stationary records~\cite{ferreira2015block} --- but practitioners are increasingly forced into data-sparse regimes where this asymptotic convergence fails, leading to volatile parameter estimates and unacceptable uncertainty in risk quantification~\cite{haan2006extreme, katz2002statistics}.

A promising alternative is the Metastatistical Extreme Value (MEV) framework~\cite{marani2015metastatistical, zorzetto2016emergence, miniussi2020metastatistical, marra2018metastatistical}, which estimates the distribution of block maxima by raising an estimated base CDF $\hat F_X$ to the power of the block size $n$, retaining information from all observations rather than only the maxima, and which has since been extended to a computationally lighter ``simplified'' form (SMEV) and applied at continental and global scale~\cite{miniussi2021estimation, grundemann2023extreme}.
Existing MEV implementations, however, rely on parametric assumptions for the base distribution (typically a Weibull, motivated by the non-asymptotic statistics of ordinary rainfall events~\cite{miniussi2020nonasymptotic}, or a data-driven choice among candidate ordinary distributions~\cite{mushtaq2022reliable}), limiting their applicability when the true tail shape is unknown.
The simplified MEV (SMEV) reduces the metastatistical model to a single ordinary distribution fitted across all blocks, lowering parameter count and estimation variance, but it retains a parametric form for the base distribution and so inherits the same vulnerability when the true tail shape is unknown or varies across blocks~\cite{miniussi2021estimation}.
The DDEVD instead reconstructs the base distribution non-parametrically, so the choice of ordinary family --- Weibull, gamma, or otherwise --- is never made and tail misspecification cannot enter through that channel.
This trades the low variance of a correctly specified SMEV for robustness against precisely the misspecification that motivates a non-parametric approach, while retaining an explicit MISE-based stability theory that purely empirical tail estimators lack.

Here we introduce the Data Driven Extreme Value Distribution (DDEVD), a non-parametric estimator designed to resolve the small-sample paradox.
The DDEVD extends MEV by replacing the parametric base distribution with a kernel density estimate~\cite{silverman2018density}, placing it in the tradition of smooth, non-parametric estimation of tail functionals~\cite{drees1998smooth} while inheriting the data efficiency of the metastatistical approach and removing distributional assumptions on the tail.
Unlike block based parametric methods that discard most observations to fit an asymptotic tail (e.g.\ retaining only annual maxima), the DDEVD uses the entire empirical distribution and recovers the tail shape directly from the data with no \textit{a priori} assumption beyond the kernel, giving it an additional layer of robustness in the face of tail misspecification.

We demonstrate the framework on two contrasting applications.
Precipitation extremes are well-approximated by Weibull statistics at the daily scale, while the stochastic behaviour of within-day maxima of 5-minute precipitaion sums --- the primary drivers of flash floods --- is governed by convective and orographic processes that defy simple parametric modelling~\cite{lenderink2008increase, ban2015heavy}.
Further, high-frequency records typically span only a few decades.
The same statistical problem arises in spatial form in the microstructural quality control of sintered refractory metals such as molybdenum, where ultimate mechanical properties are governed by the largest grains anywhere in the bulk volume, yet only a small number of two-dimensional micrographs from polished sections are available, each covering a fraction of a square millimetre~\cite{astm_e930}.
We show that DDEVD recovers stable risk curves in both domains where standard parametric fits diverge --- the same estimator working on temporal precipitation series and on spatially distributed grain-size data without inherent ordering.

Underpinning these empirical results is a rigorous derivation of the estimator's Mean Integrated Squared Error (MISE) and the resulting optimal bandwidth.
We further prove a stability criterion that links the estimator's reliable operating regime to the extreme value index $\gamma$ of the base distribution, yielding the explicit scaling law $m < C\,n^{1+\gamma/2}$, where $m$ is the number of blocks and $n$ the per-block size.
The result delineates a transparent boundary within which the DDEVD extrapolates reliably beyond the range of the data, in any of the three classical extreme-value domains.

\section*{Results}\label{sec:results}

\subsection*{Deciphering the risk of flash floods from sub-hourly records}

In the Alpine region, while daily rainfall totals are often adequately modelled by Weibull statistics, the stochastic behaviour of sub-hourly extremes --- specifically 5-minute intensity sums --- remains a distributional blind spot. Throughout, 'sub-hourly intensity' denotes the daily maximum of 5-minute precipitation sums. The annual block thus contains one such daily maximum per wet day.
These rapid bursts are the primary drivers of flash floods and urban drainage failure, yet they defy simple parametric modelling due to the complex interplay of convective instability and orography \cite{katz2002statistics}.
The geophysical processes governing these extreme events are highly localised and temporally variable, leading to tail behaviours that often deviate from standard assumptions.
For example, two stations in close proximity and at the same elevation, but separated by a mountain ridge, can exhibit fundamentally different extreme-value characteristics.
Furthermore, high-frequency 5-minute recording stations are a relatively recent addition to the monitoring network, leaving engineers with historically shallow datasets that are insufficient for classical asymptotic extrapolation.
Adding to all these problems, climate change is altering precipitation regimes, especially on the extreme ends, rendering long-term historical records less relevant for future risk assessment \cite{katz2002statistics}.

A network of semi-automated weather stations operated by Geosphere Austria (TAWES) provides high-resolution precipitation data across the Austrian Alps.
These stations record rainfall at 5-minute intervals, capturing the fine temporal structure of convective events \cite{geosphere2024klima10min, geosphere2024klimaday} and are located throughout the region.
They are a relatively recent addition to the monitoring network, being founded in 1992, resulting in at most 33 years of data per station.
Further, not every station has been equipped with 5-minute recording capabilities from the start, leading to even shorter effective records.

Unlike daily aggregates, the tail behaviour of these short-interval sums is not known \textit{a priori} and often exhibits tails inconsistent with standard Weibull assumptions: at a representative long-record station the Weibull fit assumed by the MEV baseline is adequate for daily totals but is clearly rejected for the 5-min intensities (Extended Data Fig.~\ref{edfig:raindist}).
We conducted a rigorous ``hindcast'' experiment to simulate the data-sparse conditions encountered when only short records are available or to account for temporal variability due to climate change.
We isolated short, 10-year windows from long-term station records and attempted to predict the 100-year Return Level ($RL_{100}$) using both the classical Generalized Extreme Value (GEV) distribution (via Block Maxima) and the DDEVD.

\begin{figure}[h]%
\centering
\includegraphics[width=0.98\textwidth]{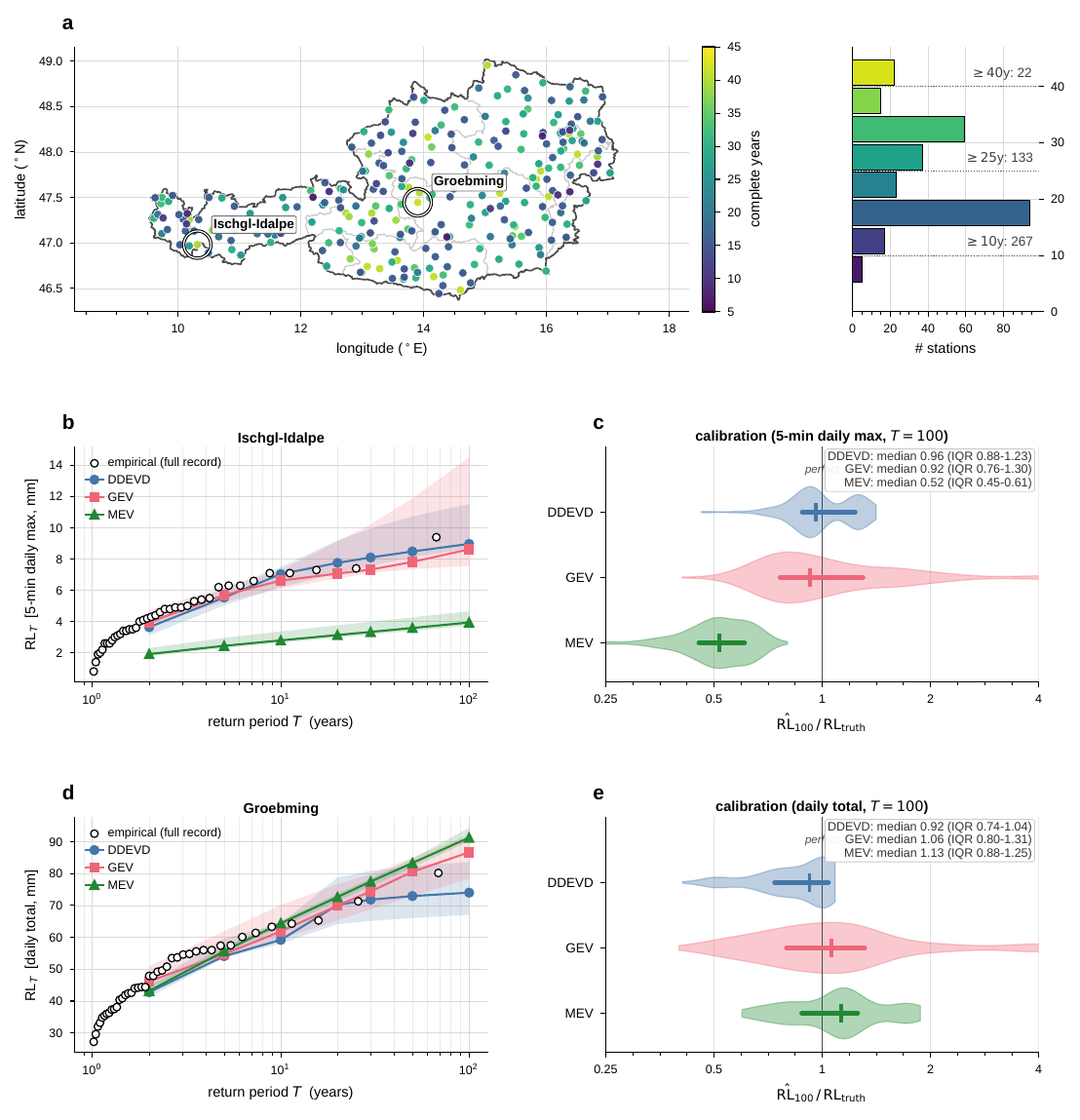}
\caption{\textbf{DDEVD return-level estimates are calibrated against the full-record truth across station archetypes and observation timescales, whereas the meteorological standard MEV under-predicts sub-hourly extremes and GEV is overdispersed.}
\textbf{a}, Spatial distribution of the $n = 378$ Austrian TAWES stations, coloured by record length (complete years). Right histogram: stations per 5-year record-length bin. Dotted lines mark the cumulative number of stations with $\geq 10$, $\geq 25$ and $\geq 40$ complete years. Black ring markers indicate the two representative stations: Ischgl-Idalpe (high-alpine) and Gr\"obming (mid-elevation).
\textbf{b}, Return-level curves at Ischgl-Idalpe for the 5-min daily maximum intensity ($\mathrm{RL}_T$, $\mathrm{mm}$). Curves are per-method medians across all 10-year hindcast windows. Shaded ribbons mark the $25\text{--}75\,\%$ inter-window range. DDEVD (blue), GEV (red), MEV with Weibull base (green). Open circles, Cunnane plotting positions from the full-record annual maxima.
\textbf{c}, Pooled calibration of 100-year return-level estimates $\widehat{\mathrm{RL}}_{100}/\mathrm{RL}_{\mathrm{truth}}$ as horizontal split-violins (median tick, $25\text{--}75\,\%$ IQR bar). DDEVD: $0.96$ ($0.88\text{--}1.23$); GEV: $0.92$ ($0.76\text{--}1.30$); MEV: $0.52$ ($0.45\text{--}0.61$) -- MEV under-predicts by ${\sim}50\,\%$.
\textbf{d, e}, As \textbf{b, c}, for daily total precipitation at Gr\"obming. DDEVD: $0.92$ ($0.74\text{--}1.04$); GEV: $1.06$ ($0.80\text{--}1.31$); MEV: $1.13$ ($0.88\text{--}1.25$). All three methods are well-calibrated on daily totals. DDEVD additionally calibrates on sub-hourly intensities where the parametric baselines fail. Hindcast protocol and reference definition in Methods.}
\label{fig:rainfall}
\end{figure}

As shown in Figure~\ref{fig:rainfall}, the classical parametric approach using GEV yielded 100-year return levels for 5-minute intensities with a very high variance.
The GEV estimate departed by more than $50\%$ from the full-record reference in roughly one in five 10-year windows. The DDEVD departed by that much in fewer than one in fifty.
The classical MEV estimator with a Weibull base distribution has been added to show that the parametric assumption is the main source of error: while MEV is well-calibrated on daily totals, it under-predicts 5-minute intensities by ${\sim}50\,\%$ at the 100-year event, a consequence of the Weibull tail being too light to capture the convective bursts that drive flash floods. Also, the shape of the distribution is not well captured by the Weibull distribution leading to wrong parameter estimates and thus wrong return level estimates, cf. Extended Data Fig.~\ref{edfig:raindist}. 
In principle this could be remedied by fitting a different parametric family, but the true tail shape is unknown and may vary across stations and time periods, making a non-parametric approach more robust.
By treating the per-wet-day maximum 5-minute sums as a metastatistical population --- retaining one maximum observation per wet day rather than only the annual maximum --- the estimator identified the underlying risk curve invisible to standard methods.
This effectively bridges the gap between modern high-frequency monitoring and century-scale safety requirements.

\subsection*{Estimating the largest grain from microstructural snapshots}

Refractory metals such as molybdenum, used across medical technologies including rotating X-ray anodes, radiation-therapy collimators, and structural implant components, are routinely processed by powder metallurgy.
The resulting polycrystalline microstructure is heterogeneous: complex thermomechanical transport during sintering produces grain populations whose tail --- the largest grains --- governs mechanical performance under cyclic thermal or mechanical load.
Fatigue crack initiation in such materials nucleates preferentially at the largest grains, making the upper tail of the grain-size distribution, rather than its bulk, the safety-relevant quantity.

The norm ASTM E930-18~\cite{astm_e930} recognizes this by noting that large grains present in the material can lead to premature failure and hence prescribes grain size limits that need to be met for certification.

Because each micrograph covers only a fraction of a square millimetre, the largest grain in a full component is necessarily larger than any grain captured in a single field of view: the safety-relevant tail must be \emph{extrapolated} from a limited area rather than observed directly.
This is an extreme-value extrapolation problem, and its reliability hinges on how the upper tail of the grain-size distribution is modelled.

Standard metallurgical practice either just takes the observed maximum grain size in good faith or (if more rigor is desired) assumes grain sizes follow a log-normal distribution and reads the tail off this bulk fit.
The assumption is convenient but unverified: a distribution calibrated on the central grain population carries no guarantee in the tail, and for new materials or processing routes the true tail shape is unknown --- for the present molybdenum specimen, log-normality is, although convenient, not the correct model (Extended Data Fig.~\ref{edfig:graindist}).

We therefore applied the DDEVD to SAM-segmented micrographs to estimate the largest grain expected over a given material area directly from the data, without assuming a tail shape, and benchmarked it against both the bulk log-normal extrapolation and a rigorous generalised extreme-value (block-maxima) fit.

\begin{figure}[h]%
\centering
\includegraphics[width=0.95\textwidth]{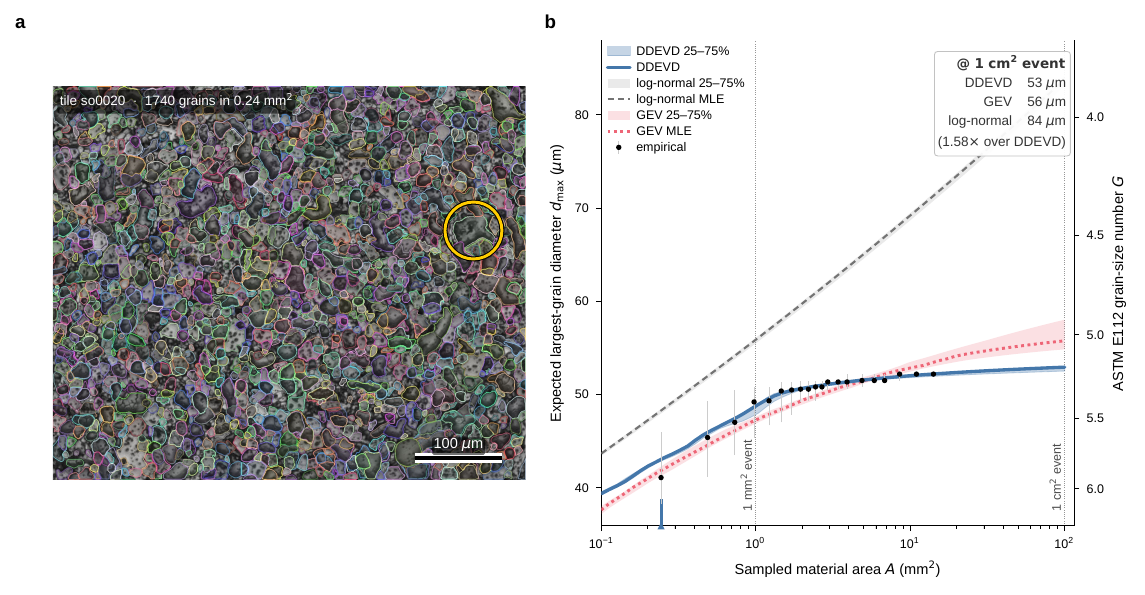}
\caption{\textbf{DDEVD applied to optical-metallography data agrees with the GEV block-maxima estimator on the expected largest grain in a given material area, and disagrees with the field-standard log-normal extrapolation by $58\,\%$ at the $1\,\mathrm{cm}^{2}$ event.}
\textbf{a}, Representative SAM-vit-l segmentation overlay of one optical micrograph of the polycrystalline microstructure ($1\,740$ valid grains in $0.24\,\mathrm{mm}^{2}$). Each coloured region is a single segmented grain, segmentation by Segment-Anything-Model vit-l (Meta AI)~\cite{kirillov2023segment} with default parameters, post-filtered to $100 \leq \mathrm{area} \leq 10^{5}\,\mathrm{px}^{2}$. Scale bar, $100\,\mu\mathrm{m}$.
\textbf{b}, Predicted median diameter $d_{\max}$ of the largest grain (equivalent circular diameter, $\mu\mathrm{m}$) as a function of sampled material area $A$ ($\mathrm{mm}^{2}$, log scale).
\emph{Blue solid line and band:} DDEVD median and $25\text{--}75\,\%$ range over $51$ fits ($20$ leave-one-tile-out jackknife $+$ $30$ half-data bootstrap resamples of the $58$ tiles, plus the central all-tiles fit; $n_{\mathrm{tot}} = 104\,321$ grains).
\emph{Grey dashed:} log-normal MLE to pooled grain areas.
\emph{Red dotted:} GEV MLE on the $58$ tile-level maxima of $\ln(\mathrm{area})$, with $k$-block scaling at $k = A/A_{\mathrm{tile}}$; fitted shape $\widehat{\xi} = -0.32$ (Weibull domain, finite endpoint $\sim 58\,\mu\mathrm{m}$).
\emph{Black points:} empirical block-pooled maxima (median, whiskers $25\text{--}75$ percentile; 200 random partitions per pool size).
Blue triangle on the abscissa: natural block area ($0.24\,\mathrm{mm}^{2}$), data-supported region extends to $14.1\,\mathrm{mm}^{2}$. Vertical guides at $A = 1\,\mathrm{mm}^{2}$ and $1\,\mathrm{cm}^{2}$, right ordinate gives the ASTM~E112 grain-size number $G$. At $1\,\mathrm{cm}^{2}$ the two EVS estimators agree (DDEVD $53\,\mu\mathrm{m}$, GEV $56\,\mu\mathrm{m}$; $G \approx 5.1$), while the bulk-fitted log-normal extrapolates to $84\,\mu\mathrm{m}$ (a $1.58$-fold over-prediction; log-normality rejected, Anderson--Darling $A^{2} = 171.3$, $\alpha < 0.01$). Fitting protocols and hyperparameters in Methods (SI, Section~\ref{sec:SI_methods_metallurgy}).}
\label{fig:metal}
\end{figure}

As shown in Fig.~\ref{fig:metal}, the bulk-fitted log-normal model---while adequate for the central grain population---does not reproduce the upper tail that governs mechanical performance.
Extrapolated to the $1\,\mathrm{cm}^{2}$ event, it predicts a largest grain of $84\,\mu\mathrm{m}$, a $1.58$-fold over-prediction relative to the DDEVD ($53\,\mu\mathrm{m}$).
The DDEVD recovers this tail directly from the data without imposing a distributional form, and agrees to within ${\sim}5\,\%$ with an independent generalised extreme-value fit to the tile-level block maxima ($56\,\mu\mathrm{m}$; shape $\widehat{\xi} = -0.32$, a bounded Weibull-domain tail).
Although the two estimators agree on the median, they differ sharply in precision under extrapolation: at the $1\,\mathrm{cm}^{2}$ event the DDEVD $25\text{--}75\,\%$ inter-quartile band spans only $0.4\,\mu\mathrm{m}$ ($0.8\,\%$ of the median) against $3.2\,\mu\mathrm{m}$ ($5.7\,\%$) for the GEV --- an $87\,\%$ narrower band (Fig.~\ref{fig:metal}b). Within the data-supported region (e.g.\ $1\,\mathrm{mm}^{2}$) the two bands are comparable; the DDEVD advantage emerges precisely where it matters --- extrapolating beyond the observed area --- because the GEV must propagate the uncertainty of a shape parameter estimated from only $58$ tile maxima, whereas the DDEVD constrains the tail using all ${\sim}10^{5}$ grains.
The convergence of two independent extreme-value estimators---one metastatistical, one block-maxima---on the same risk curve, where the bulk log-normal diverges by more than $50\,\%$, indicates that the fatigue-relevant grain-size tail is materially lighter than a bulk fit implies: the largest grains are bounded rather than log-normally heavy.
Because the field-standard log-normal is calibrated on the bulk rather than the extremes, it mis-estimates precisely the outlier grains that act as stress concentrators and initiate fatigue cracks under cyclic loading.
The DDEVD thus provides a non-parametric, EVT-consistent route to this tail from standard quality-control micrographs, reducing reliance on large-scale destructive testing.

\subsection*{Asymptotic stability and the limits of extrapolation}

The success of the DDEVD in these disparate regimes --- atmospheric physics and metallurgic imaging --- is not accidental.
It rests on a fundamental stability criterion we derived by optimizing the estimator's Mean Integrated Squared Error (MISE).
For the exact definitions of all occurring quantities and the full derivation see the SI, for special cases we will reference the relevant sections.

The key quantity is the positive definiteness of the Hessian matrix $\mathbf{Q}$ of the MISE with respect to the bandwidth vector.
We show (SI, Section~\ref{sec:SI_stability}) that $\mathbf{Q}$ is positive definite if and only if a scalar functional $D$, which depends on the tail behaviour of $F_X$, is positive.
Asymptotic analysis of $D$ for large block sizes $n$ yields the scaling law
\begin{equation}\label{eq:stability_main}
  m < C\,n^{1+\gamma/2},
\end{equation}
where $\gamma$ is the extreme value index of $F_X$~\cite{haan2006extreme,ferreira2015block} and $C$ is a positive constant depending on $F_X$ and the kernel $K$.
This result holds for all distributions in the Fr\'{e}chet ($\gamma > 0$) and short-tailed Weibull ($-1/2 < \gamma < 0$) domains of attraction, and extends to the Gumbel case ($\gamma = 0$) up to logarithmic corrections that arise from the special asymptotic representation of distributions in that domain (see SI, Section~\ref{sec:SI_stability}).

\begin{figure}[h]
\centering
\includegraphics[width=0.95\textwidth]{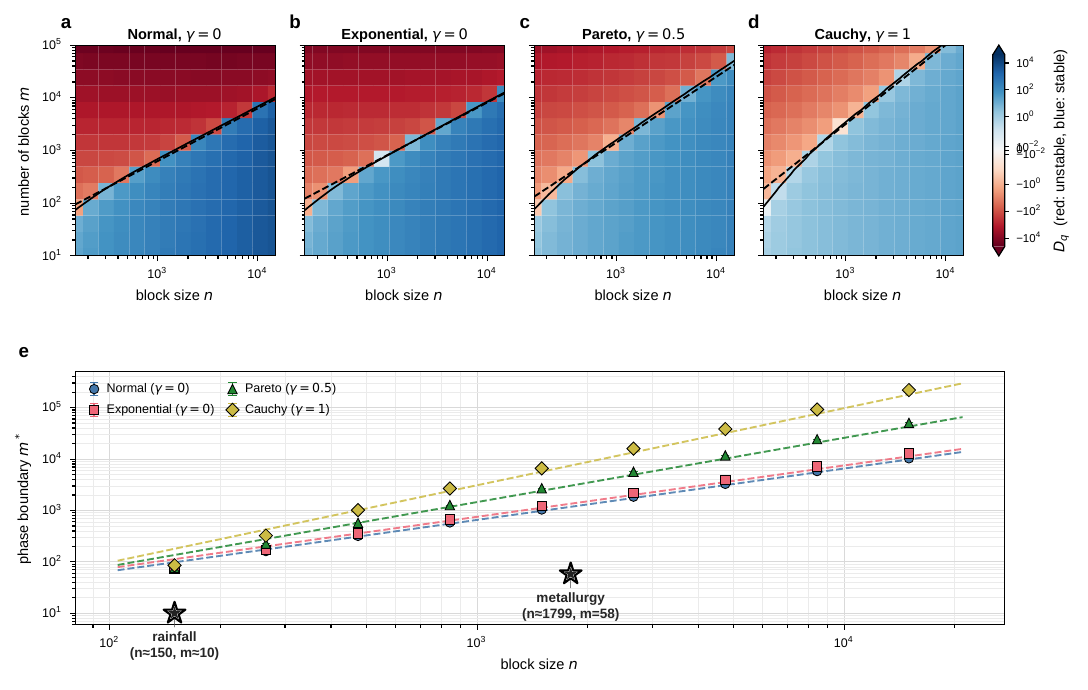}
\caption{\textbf{The DDEVD plug-in bandwidth optimiser has an empirically verifiable stability boundary that holds across the three extreme-value domains, and both real-world applications operate well inside it.}
\textbf{a--d}, Phase diagrams of the smallest eigenvalue $D_q$ of the plug-in Hessian on a log-spaced grid of block size $n$ (abscissa) and number of blocks $m$ (ordinate), for four reference distributions: \textbf{a}, Normal ($\gamma = 0$, Gumbel); \textbf{b}, Exponential ($\gamma = 0$, Gumbel); \textbf{c}, Pareto $\alpha = 2$ ($\gamma = 0.5$, Fr\'echet); \textbf{d}, Cauchy ($\gamma = 1$, Fr\'echet). Blue cells, $D_q > 0$ (stable); red cells, $D_q \leq 0$ (unstable). The solid black contour is the empirical $D_q = 0$ isoline; the dashed line is the theoretical boundary $m = C \cdot n^{1+\gamma/2}$ (SI, Section~\ref{sec:SI_stability}), with $C$ fit by least squares over the empirical boundary points. The two boundaries agree to within grid resolution across all four distributions.
\textbf{e}, Pooled view: empirical boundary points and theoretical scaling lines for all four distributions on a single log-log axis. Slopes $1+\gamma/2$ are recovered: $1.00$ for Normal/Exponential, $1.25$ for Pareto, $1.50$ for Cauchy. Black stars indicate the operating regimes of the two applications (Figs.~\ref{fig:rainfall},~\ref{fig:metal}): rainfall ($n \approx 150$, $m \approx 10$) and metallurgy ($n \approx 1\,799$, $m = 58$); both sit at least an order of magnitude below the stability boundary. Synthetic-grid protocol in Methods (SI, Section~\ref{sec:SI_methods_stability}).}
\label{fig:stability}
\end{figure}

This result provides a rigorous phase diagram for the estimator's applicability (Fig.~\ref{fig:stability}).
Crucially, both precipitation extremes ($\gamma > 0$) and microstructural anomalies fall within this stable region for typical observation windows.
For the precipitation application, with blocks of $n \approx 150$ rainy days per year and $m = 10$ years per hindcast window, the stability criterion yields a comfortable margin.
For the metallurgical data, with $\bar n \approx 1\,799$ grains per micrograph and $m = 58$ fields of view, the criterion is likewise satisfied.
This explicit stability criterion distinguishes the DDEVD from heuristic ``black box'' methods, offering a transparent boundary within which the method extrapolates reliably beyond the range of the data.

\section*{Discussion}\label{sec:discussion}

Classical EVT discards the vast majority of observations to fit asymptotic tail models, an approach that is well-founded asymptotically but problematic in the data-sparse regimes that are increasingly common in engineering practice.
The DDEVD addresses this by replacing asymptotic tail fitting with a kernel-based reconstruction of the full base distribution, inheriting the data efficiency of the metastatistical approach while avoiding parametric assumptions about tail shape.

In Alpine hydrology, the ability to derive stable 100-year return levels from a single decade of 5-minute records has immediate engineering value: flood-defence standards require return-level estimates that reflect the current climate, yet the effective stationarity window for sub-hourly precipitation is often shorter than a decade~\cite{katz2002statistics, westra2014future}; recent work extends the metastatistical framework to explicitly non-stationary and seasonal regimes~\cite{vidriosahagun2022hydrological, vidriosahagun2023nonstationary, falkensteiner2023accounting}. The DDEVD narrows this gap.
Building on the station-wise analysis presented here, a method for the spatial interpolation of 5-minute and other sub-daily precipitation extremes across Austria --- enabling return-level estimation at ungauged locations --- is in preparation by the authors.

In materials science, qualification of advanced materials by additive manufacturing or powder metallurgy routinely involves microstructures whose grain-size tails are not well described by the standard assumption of log-normal statistics.
The DDEVD offers a non-parametric pathway to characterise these tails from standard quality-control micrographs, reducing reliance on large-scale destructive testing.

Two limitations deserve mention, in keeping with recent critical appraisals of the metastatistical approach regarding bias, uncertainty bands, and goodness-of-fit~\cite{schmith2026critical}. The DDEVD is not a panacea: the kernel density estimator converges slowly in the far tail, and the bandwidth optimisation itself has a theoretical stability boundary.
The derived criterion ($m < C \cdot n^{1+\gamma/2}$) makes this boundary explicit and calculable, but for distributions with very large $\gamma$ or small sample sizes additional caution is warranted.
Additionally, if the mass of the distribution is concentrated at a finite point, but the tail is sufficiently heavy, then, although the data lies in the stable region, a transformation of the data may be required to recover a usable estimate of the tail.
We will discuss this further in the following section.
A second open question is the effect of temporal or spatial dependence within blocks on the MISE expansion. The present theory assumes block-internal independence, an assumption that holds well for annual precipitation maxima but may require scrutiny for shorter blocks or strongly correlated spatial data.

More broadly, the DDEVD framework is applicable wherever data can be organised into approximately exchangeable blocks and the block-maximum distribution is of primary interest.
The theoretical stability criterion derived here provides a practical guideline for assessing whether a given dataset falls within the reliable operating regime of the estimator.

\section*{Methods}\label{sec:methods}

For detailed definitions and derivations of objects in this section, see the Supplementary Information (SI).
For an even more detailed exposition, see the accompanying technical report~\cite{sandbichler2026ddevd_companion}.

\subsection*{The DDEVD Estimator}
Let $X_{ij}$ denote the $j$-th sample within the $i$-th block, where $i=1,\dots,m$ and $j=1,\dots,n_i$.

We define the Data Driven Extreme Value Distribution (DDEVD) estimator $\hat{F}_{h}(y)$ as the metastatistical aggregation of kernel-smoothed block distributions
\begin{equation}
\hat{F}_{h}(y) = \frac{1}{m}\sum_{i=1}^{m} \left( \frac{1}{n_{i}}\sum_{j=1}^{n_{i}} K\left(\frac{y-X_{ij}}{h_{i}}\right) \right)^{n_{i}}
\end{equation}
where $K(u) = \int_{-\infty}^u k(t)\,dt$ is a smooth cumulative distribution function (CDF) kernel derived from a symmetric, zero-mean probability density $k(t)$ (e.g., Gaussian), and $h = (h_1, \dots, h_m)$ is the vector of block-specific bandwidths.
The asymptotic analysis assumes that the base distribution $F_X$ is twice continuously differentiable on the support of interest and that the kernel density $k$ together with $[K^2]'$ admit finite first and second moments.
The standard Gaussian kernel used throughout this work satisfies these conditions, but other kernels meeting these criteria, such as a centered Weibull kernel, could also be used.

Unlike classical Block Maxima approaches which reduce the dataset to a set of $m$ maxima, this estimator retains the information from all $\sum_{i=1}^m n_i$ observations.
The selection of the bandwidth vector targets the quantile-restricted Mean Integrated Squared Error
\begin{equation}
  \mathrm{MISE}_q(\hat{F}_h) = \E \int_{F^{-1}_X(q)}^{\infty} \bigl(\hat{F}_h(y) - F(y)\bigr)^2\,dy,
\end{equation}
which concentrates the loss on the upper tail relevant for risk assessment for $0<q<1$.

\subsection*{Optimization of the Bandwidth Vector}
The selection of the bandwidth vector $h$ is critical for the stability of the exponentiated estimator.
We derived the asymptotic Mean Integrated Squared Error (MISE) of $\hat{F}_{h}$ up to order $\mathcal{O}(h^3)$ and obtained the optimal bandwidth vector $h_{opt}$ as the solution to the linear system $Q h_{opt} = -\frac{1}{2} c$,
where $Q$ is the Hessian matrix of the bias-variance functional and $c$ is the coefficient vector.

Since $Q$ and $c$ depend on the unknown underlying density $f_X$ and its derivative, we implement a fixed-point iterative plug-in algorithm

\begin{enumerate}
    \item \textbf{Initialization:} We set the initial bandwidth $h^{(0)}$ using Silverman's rule of thumb scaled for the block size, $h_{i}^{(0)} = 1.06\,\hat{\sigma}_i\, n_i^{-1/5}$~\cite{silverman2018density}.
    \item \textbf{Pilot Estimation:} At iteration $l$, we construct pilot density and CDF estimates from the current bandwidths,
    \begin{align}
    \hat{f}^{(l)}(y) &= \frac{1}{m}\sum_{i=1}^{m}\frac{1}{n_i\,h_i^{(l)}}\sum_{j=1}^{n_i} k\!\left(\frac{y-X_{ij}}{h_i^{(l)}}\right),\\
    \hat{F}^{(l)}(y) &= \frac{1}{m}\sum_{i=1}^{m}\frac{1}{n_i}\sum_{j=1}^{n_i} K\!\left(\frac{y-X_{ij}}{h_i^{(l)}}\right),
    \end{align}
    and substitute them into the integrands defining $\mathbf{c}$ and $\mathbf{Q}$ to obtain the plug-in estimates $\hat{\mathbf{c}}^{(l)}$ and $\hat{\mathbf{Q}}^{(l)}$ via numerical integration.
    \item \textbf{Update:} We compute the raw update $\mathbf{h}^* = -\frac{1}{2}\,(\hat{\mathbf{Q}}^{(l)})^{-1}\hat{\mathbf{c}}^{(l)}$ and apply a damped update rule $\mathbf{h}^{(l+1)} = (1-\lambda)\,\mathbf{h}^{(l)} + \lambda\, \mathbf{h}^*$ with relaxation parameter $\lambda = 0.5$.
    The damping provides regularisation against numerical oscillation in cases where the update map from $\mathbf{h}^{(l)}$ to $\mathbf{h}^*$ is not a contraction in the pre-asymptotic regime.
    \item \textbf{Termination:} The procedure repeats until the relative change in bandwidth satisfies $\|\mathbf{h}^{(l+1)} - \mathbf{h}^{(l)}\|_2 / \|\mathbf{h}^{(l)}\|_2 < \varepsilon$, with $\varepsilon$ chosen typically at $10^{-4}$.
\end{enumerate}

Experiments suggest that the plug-in algorithm converges reliably to the optimal bandwidth when the stability criterion is satisfied, and that the resulting DDEVD estimates are robust to the choice of initial bandwidth and relaxation parameter within reasonable ranges, however a rigorous proof of convergence is an open question~\cite{sandbichler2026ddevd_companion}.

\noindent\textbf{Per-group bandwidth solve.}
The bandwidth optimisation step has two practical limitations at large $m$: a direct solve of $\hat{\mathbf{Q}}\mathbf{h}^{*} = -\tfrac{1}{2}\hat{\mathbf{c}}$ is dominated by the $O(m^{2})$ cost of constructing $\hat{\mathbf{Q}}$ entry-wise via numerical integration, and at small $n$ the global operating point $(n, m)$ can cross the stability boundary of Eq.~\eqref{eq:stability_main} so that $\hat{\mathbf{Q}}$ fails to be positive definite.
We address both with a single documented approximation: when $m$ exceeds a user-set threshold $M$ (parameter \texttt{max\_bins}, default $10$ in the metallurgy application), the $m$ blocks are partitioned into $\lceil m / M \rceil$ contiguous groups, the iterative plug-in is run independently on each group, and the resulting per-block bandwidths are concatenated.
The per-group operating point is $(n, m_{\mathrm{group}})$ with $m_{\mathrm{group}} \leq M$, which sits well inside the stability boundary for every $n$ encountered in the applications --- the partitioning therefore ensures positive-definite per-group Hessians by construction, even when the global $(n, m)$ would not.
The concatenated bandwidth vector $\hat{\mathbf{h}}$ is the asymptotic MISE minimiser of the per-group sub-problems.
At finite $M$ the per-group bandwidth is a pre-asymptotic approximation: a direct comparison on the metallurgy grain-area data itself at the application setting ($m = 58$, mean $\bar n = 1\,799$, $M = 10$) gives $\hat{h}_{\mathrm{per\text{-}group}} \,/\, \hat{h}_{\mathrm{global}} = 0.713$ ($\hat{h}_{\mathrm{per\text{-}group}} = 0.0203$, $\hat{h}_{\mathrm{global}} = 0.0284$ on the log-transformed scale), approaching unity monotonically from below as $M \to m$, see Fig.~\ref{edfig:max_bins_comparison}.
The MISE surface is quadratic about its minimum so the prediction is much less sensitive to this bandwidth deviation than the ratio suggests, but the practitioner should treat the partitioned solve as a finite-sample approximation rather than an asymptotic exact match. 
The DDEVD estimator itself still averages over all $m$ blocks in the MEV form, only the bandwidth optimisation step is partitioned.
The partitioned solve is used uniformly for the metallurgy application (where $m$ ranges up to $\sim 10^{3}$ at high split factors), for the rainfall application $m = 10$ so the global solve is used directly.

Hence the per-group solve tends to under-estimate the bandwidth relative to the global solve, which is consistent with the intuition that a smaller number of blocks $m_{\mathrm{group}}$ leads to a more variable pilot density and thus a more aggressive (smaller) bandwidth.
In principle, the per-group solve could be improved by a correction factor that accounts for the difference in $m$ between the per-group and global settings, but we have not found this to be necessary in practice: the per-group solve is stable and gives good results on the metallurgy data without further tuning, and the rainfall application does not require it at all since $m = 10$ is below the threshold. Finding a quantitative (asymptotic or finite-sample) correction factor is an interesting open question for future work.

Detailed derivations of the Hessian matrix entries and stability conditions are provided in SI Sections~\ref{sec:SI_mise}--\ref{sec:SI_coefficients}.

\subsection*{Stability Criterion}
A key theoretical result of our analysis is the derivation of a stability criterion for the bandwidth optimisation.
For a distribution with extreme value index $\gamma > -1/2$, the optimisation is asymptotically stable provided the number of blocks $m$ scales such that $m < C\,n^{1+\gamma/2}$.
This criterion is derived for blocks of equal size $n$. For unequal block sizes, a perturbation argument (Weyl's inequality) suggests that the criterion continues to hold with $n$ replaced by the mean block size $\bar n$, provided the relative standard deviation $\sigma_n/\bar n$ is small. We state this as Conjecture~2.1 in the companion report~\cite{sandbichler2026ddevd_companion} (see also Supplementary Information Section \ref{sec:SI_stability}).
This is the regime relevant to both applications: yearly blocks of 5-minute rainfall events have moderate inter-annual variability in block size, and the metallurgy tiles have comparable grain counts after the artefact-removal filter.
Both datasets satisfy the criterion at $\gamma > -1/2$ --- the precipitation records with an empirical $\gamma \approx 0.1\text{--}0.2$ from the rainfall hindcast and the metallurgy grain-area distribution with $\widehat\xi = -0.32$ from the GEV fit on tile-level maxima (Fig.~\ref{fig:metal}) --- placing the sample sizes used here well inside the stable region.
The proof of this criterion is detailed in SI Section~\ref{sec:SI_stability}.

\subsection*{Data transformation and diagnostics}
The kernel CDF estimator assumes a base density that is well-conditioned on the working scale (see~\cite{sandbichler2026ddevd_companion}, Section~3, for the full transformation framework and diagnostics). Sub-hourly precipitation intensities are not: the wet-event distribution combines a near-singular concentration of probability just above the detection threshold ($\sim 0.1\,\mathrm{mm}$) with a heavy upper tail several orders of magnitude away (Extended Data Fig.~\ref{edfig:raindist}a,c). A single-bandwidth fit on the linear scale then either undersmooths the bulk and leaves the q-tail as a staircase, or oversmooths and erases the tail structure---a failure mode well documented in the kernel density estimation literature for boundary-spike distributions on positive support~\cite{wand1991transformations, park1990comparison,sheather2004density}.
The standard remedy is to fit on a transformed scale where the density is better-conditioned and back-transform the result: for a strictly monotone $C^2$ map $T:\mathcal{Y}\to\mathcal{Z}$, the back-transformed estimator $\hat{F}^T_Y(y) := \hat{F}_Z(T(y))$ targets the same $F_Y$ as a direct fit and preserves the DDEVD structure. For sub-hourly rainfall we adopt $T = \log$, which maps the boundary spike to $\approx -2$ (due to the detection threshold) and the tail to $\approx 3$ (due to the largest observed events), a much more manageable range for a single bandwidth to capture.
Daily totals pass the diagnostic on the linear scale and we fit them directly ($T = \mathrm{id}$): forcing a log transform on daily totals would worsen the iterative plug-in's conditioning, leading to occasional non-convergence on short windows.
The choice is justified empirically through three diagnostic checks on the working-scale fit: convergence of the plug-in iteration, bandwidth at or above the median tail-observation spacing, and absence of multi-modal spikes in the pilot density. Failure of any of these on the linear scale and success on the log scale identifies the transformation as appropriate.

Return levels for return period $\tau$ are computed on the working scale and back-transformed, $\widehat{\mathrm{RL}}_\tau = T^{-1}\!\bigl(\hat{F}_Z^{-1}(1 - 1/\tau)\bigr)$, where the inverse $\hat{F}_Z^{-1}$ is evaluated numerically by bisection on the DDEVD CDF.

Figures~\ref{edfig:graindist} and~\ref{edfig:raindist} show the working-scale diagnostics for the metallurgy and rainfall applications, respectively. The log transform is clearly justified for the 5-minute intensities, while the daily totals are well-behaved on the original scale. The DDEVD estimates are consistent across transforms, but the log transform is necessary to recover a meaningful tail estimate for the 5-minute data.

\subsection*{Data Acquisition and Processing}

\textbf{Precipitation Data:}
High-frequency precipitation records were obtained from Geosphere Austria's semi-automated weather station network (TAWES). For the sub-hourly hindcast we selected the 21 stations in the Austrian Tyrol with at least 25 years of complete 5-minute records (out of 378 TAWES stations); for the daily-totals control we used all stations with at least 25 years of continuous daily records.
To assess estimator performance under data scarcity, we performed a hindcast analysis in which a 10-year window was slid over each station record in one-year steps; the pooled sub-hourly calibration comprises $192$ overlapping windows drawn from the six stations with at least $25$ years of 5-minute records. Within every window we computed DDEVD, MEV (Weibull base distribution fitted by MLE per block), and GEV (fitted to annual maxima by MLE) estimates of the 100-year return level on the per-year block-maxima of daily 5-minute sums and daily totals. DDEVD is fitted with \texttt{h\_opt\_position="quantile\_0.95"} (q-MISE optimised on the upper $5\,\%$ tail) and a per-variable data transform: \texttt{transform="log"} for the sub-hourly intensity series (justified by the boundary-spike diagnostic, see Section~\ref{sec:SI_methods_rainfall}) and \texttt{transform=None} (linear scale) for daily totals, which pass the diagnostic on the original scale; baseline GEV and MEV fits are reused across transforms within each window.
The empirical reference (``ground truth'') is the return level estimated from the full station record by the Cunnane plotting-position formula $p_{k}=(k-0.4)/(N+0.2)$, which makes no distributional assumption on the maxima but is itself a finite-sample estimate. To quantify this reference uncertainty we additionally compute a non-parametric bootstrap of the full-record annual maxima ($B = 500$ resamples with replacement of the $N$ available years); the $25\text{--}75\,\%$ inter-resample range of the Cunnane $\mathrm{RL}_{T}$ is reported alongside the point estimate in $\texttt{truth\_<var>.csv}$ and overlaid as a vertical band on Fig.~\ref{fig:rainfall}c,e. Calibration ratios $\widehat{\mathrm{RL}}_{T}/\mathrm{RL}_{\mathrm{truth}}$ within this band are statistically indistinguishable from the reference at the chosen confidence level.

\textbf{Metallurgical Data:}
Microstructural data were acquired by optical microscopy on a polished cross-section of a sintered, production-grade molybdenum specimen produced and provided by the materials testing laboratory of Plansee SE (Reutte, Austria). Each field of view is a tile of $1\,300 \times 1\,080$ pixels; 108 such tiles were processed.
Grain segmentation was performed using the Segment Anything Model (SAM)~\cite{kirillov2023segment}, a deep-learning-based instance segmentation approach.
Grain areas were calculated from the pixel count of segmented domains, with artefacts removed by excluding segments outside the range $[100, 10^{5}]\,\mathrm{px}^{2}$ (eliminating noise specks below the resolution limit and matrix/background segments above the physical grain scale).
After filtering and exclusion of boundary/partial tiles, $m = 58$ valid interior tiles were retained for the analysis ($n_{\mathrm{tot}} = 104\,321$ grains; per-tile mean $\bar n = 1\,799$). Each valid tile constitutes one DDEVD block, and the DDEVD configuration follows \texttt{h\_opt\_position="quantile\_0.9"}, \texttt{target\_distribution=stats.norm} (Gaussian pilot on $\ln(\mathrm{area})$), \texttt{transform="log"}, and \texttt{max\_bins=10} (per-group bandwidth solve as above). Uncertainty quantification in Fig.~\ref{fig:metal} pools two resampling modes implemented in the supporting pipeline: (i) leave-one-tile-out jackknife ($B_{\mathrm{j}} = 20$ resamples) for DDEVD; and (ii) half-data bootstrap ($B_{\mathrm{h}} = 30$ resamples drawing $\lceil m/2 \rceil$ blocks without replacement) applied uniformly to DDEVD, log-normal MLE, and GEV MLE so that all three competing methods carry an inter-resample $25\text{--}75\,\%$ band on equal footing. The empirical reference is the pooled block-maxima distribution computed from the same $58$ tiles. Full hyperparameters in Section~\ref{sec:SI_methods_metallurgy}.


\bmhead{Data availability}
The synthetic datasets and minimal working examples used for the numerical experiments are available alongside the code at \url{https://github.com/DataLabHell/ddevd}. The quality-controlled station precipitation data are openly available from GeoSphere Austria's data hub under the CC BY 4.0 licence: daily totals from the \texttt{klima-v2-1d} dataset~\cite{geosphere2024klimaday} and 10-minute values from the \texttt{klima-v2-10min} dataset~\cite{geosphere2024klima10min}. The 5-minute resolution series used for the sub-hourly analysis is not part of the open archive and was provided by GeoSphere Austria on request. The micrograph segmentations supporting the metallurgy analysis are openly available at Zenodo (\url{https://doi.org/10.5281/zenodo.20621678}) under the Apache-2.0 licence.

\bmhead{Code availability}
The Python implementation of the DDEVD estimator, including the iterative bandwidth selection and the analysis scripts that reproduce the figures, is available at \url{https://github.com/DataLabHell/ddevd}.

\bmhead{Acknowledgements}
We thank Kevin Ploner and Norbert Köpfle from the materials testing laboratory of Plansee SE for providing the molybdenum specimen and metallographic preparation. We further thank Harald Schellander and Alexander Radlherr from GeoSphere Austria for providing the precipitation data and domain expertise.

\bmhead{Author contributions}
M.S. developed the theoretical framework and derived MISE and optimal bandwidth. T.H. conceived the metastatistical application and the basic structure of the estimator.

\bmhead{Competing interests}
The extreme-value methodology developed in this work --- including the DDEVD estimator, its mean integrated squared error analysis, and the associated stability criterion --- is derived in a self-contained mathematical companion paper~\cite{sandbichler2026ddevd_companion} (Sandbichler \& Hell, \texttt{arXiv:2605.21416}) independently of any application. The present manuscript demonstrates this methodology on two empirical datasets that became available through cooperations of Data Lab Hell GmbH, which we disclose below.
Data Lab Hell GmbH conducted research on rainfall extremes within a funded cooperation with GeoSphere Austria; the precipitation application reported here uses TAWES records made available through this cooperation. Geosphere Austria provided access to the TAWES records and domain expertise but had no role in study design, data analysis, decision to publish, or manuscript preparation.
Data Lab Hell GmbH additionally maintains a research and development cooperation with Plansee SE that includes work on the analysis of metallographic micrographs. Plansee SE provided the molybdenum specimen and metallographic preparation used in this study and had no role in study design, data analysis, decision to publish, or manuscript preparation. The industrial collaboration did not influence the theoretical development of the method.
The authors declare no further competing interests.



\clearpage
\begingroup
\setcounter{figure}{0}
\renewcommand{\figurename}{Extended Data Fig.}
\renewcommand{\thefigure}{\arabic{figure}}

\begin{center}{\Large \textbf{Extended Data}}\end{center}
\bigskip

\begin{figure}[h]
\centering
\includegraphics[width=0.92\textwidth]{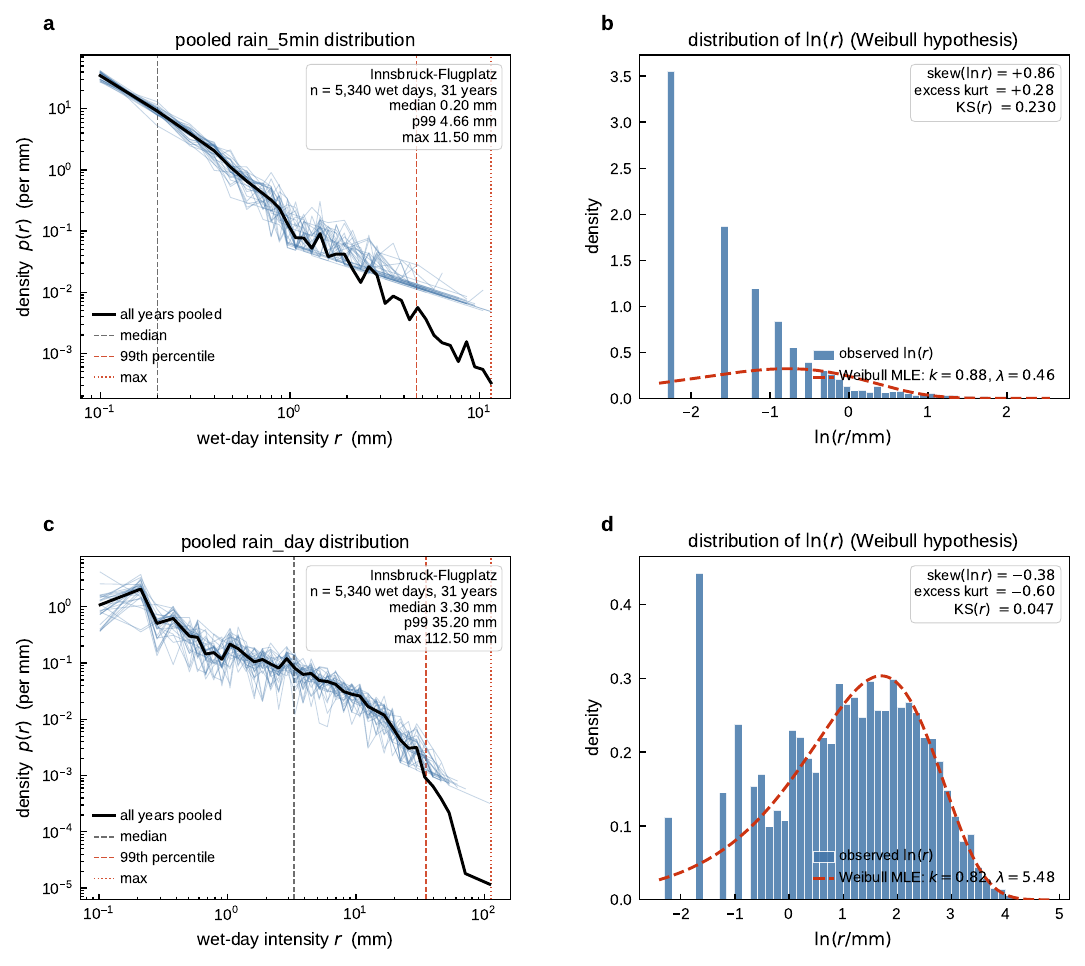}
\caption{\textbf{Empirical wet-day intensity distributions motivate a non-parametric base distribution and a logarithmic working scale.}
Pooled wet-day intensity distributions at Innsbruck-Flugplatz ($31$ complete years, $n = 5\,340$ wet days), for the 5-min daily-maximum intensity (\textbf{a, b}) and the daily total (\textbf{c, d}).
\textbf{a, c}, Pooled density on log--log axes (black) with per-year curves (thin blue); vertical lines mark the median, 99th percentile and maximum. Both variables span $2\text{--}3$ decades with probability mass concentrated just above the $\sim 0.1\,\mathrm{mm}$ gauge-resolution floor --- the boundary-spike, positive-support regime in which a single linear-scale kernel bandwidth cannot resolve the bulk and the tail simultaneously, motivating the logarithmic working scale.
\textbf{b, d}, Histogram of $\ln r$ with the back-transformed Weibull maximum-likelihood fit assumed by the MEV baseline (red dashed). The Weibull approximation is adequate for daily totals ($\mathrm{KS}(r) = 0.047$; near-symmetric $\ln r$, skew $-0.38$) but breaks down for the 5-min intensities ($\mathrm{KS}(r) = 0.230$; right-skewed, skew $+0.86$), motivating the assumption-free base distribution of the DDEVD. KS, Kolmogorov--Smirnov statistic of the Weibull fit to $r$.}
\label{edfig:raindist}
\end{figure}

\begin{figure}[h]
\centering
\includegraphics[width=0.98\textwidth]{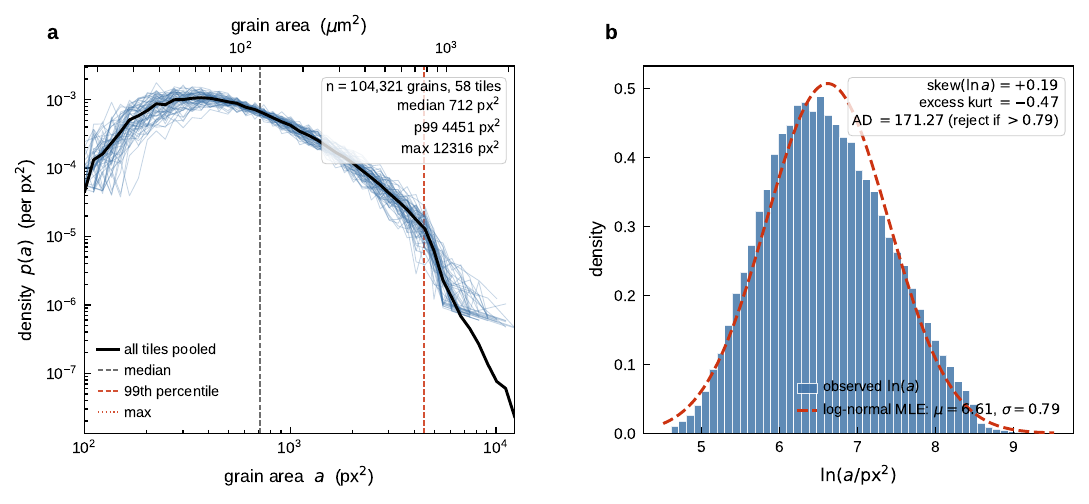}
\caption{\textbf{The pooled grain-area distribution spans several decades and is not log-normal, motivating a non-parametric tail estimate on a logarithmic scale.}
SAM-segmented grain areas pooled over the $58$ valid tiles ($n_{\mathrm{tot}} = 104\,321$ grains, after the merge-artefact blocklist of Fig.~\ref{fig:metal}).
\textbf{a}, Pooled density on log--log axes (black) with per-tile curves (thin blue); vertical lines mark the median, 99th percentile and maximum. Grain areas span ${\sim}2.5$ decades ($10^{2}\text{--}10^{4}\,\mathrm{px}^{2}$), so the DDEVD is fitted on $\ln(\mathrm{area})$. Top axis, equivalent area in $\mu\mathrm{m}^{2}$.
\textbf{b}, Histogram of $\ln(\mathrm{area})$ with the log-normal (Gaussian-in-$\ln$) maximum-likelihood fit (red dashed). Log-normality is rejected (Anderson--Darling $A^{2} = 171$ against the $\alpha = 0.05$ critical value $0.79$; residual skew $+0.19$), so the field-standard bulk log-normal cannot be trusted in the safety-relevant tail (Fig.~\ref{fig:metal}).}
\label{edfig:graindist}
\end{figure}

\begin{figure}[h]%
\centering
\includegraphics[width=0.95\textwidth]{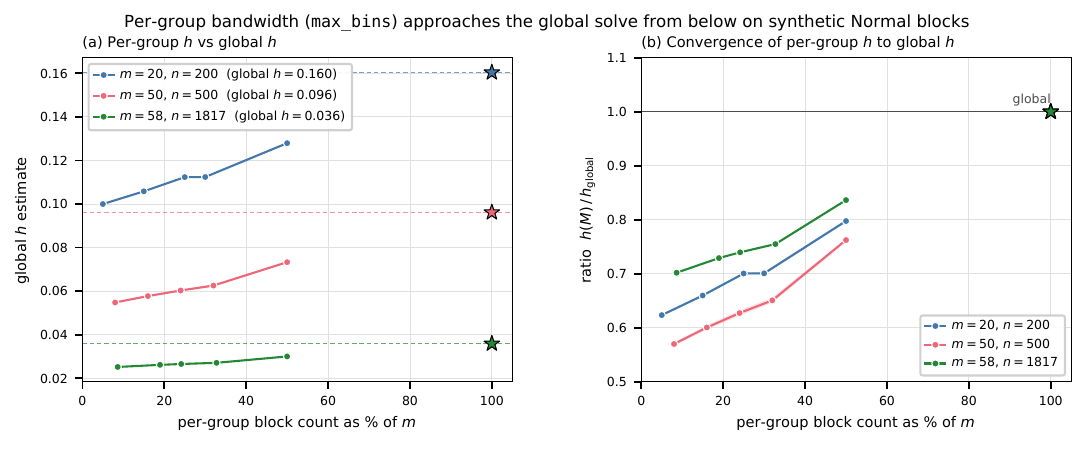}
\caption{\textbf{Per-group bandwidth approximation converges monotonically to the global solve on synthetic Normal data.}
Each curve corresponds to one $(m, n)$ configuration: $(20, 200)$, $(50, 500)$, and $(58, 1817)$ (the metallurgy application operating point); shaded bands show $\pm 1\sigma$ across three random seeds; stars at $100\,\%$ mark the single-solve global reference.
\textbf{a}, Global bandwidth estimate $\hat{h}$ as a function of per-group block count $m_{\mathrm{group}}$ expressed as a percentage of $m$, for $\mathcal{N}(0,1)$ blocks fitted without transformation. Horizontal dashed lines repeat the single-solve value for each configuration; per-group estimates approach the global reference from below as $m_{\mathrm{group}} \to m$.
\textbf{b}, Ratio $\hat{h}(m_{\mathrm{group}}) / \hat{h}_{\mathrm{global}}$; the horizontal line at unity is the global reference. The ratio increases monotonically toward unity and does not overshoot, confirming the per-group solve is a conservative (downward-biased) finite-sample approximation of the global bandwidth.}
\label{edfig:max_bins_comparison}
\end{figure}

\endgroup

\clearpage
\appendix
\renewcommand{\thesection}{S\arabic{section}}
\setcounter{section}{0}

\begin{center}
  {\Large \textbf{Supplementary Information}}\\[6pt]
  {\large Robust estimation of extreme event risks from sparse data\\
  in climate and materials science}\\[6pt]
  Michael Sandbichler, Tobias Hell
\end{center}

\tableofcontents
\bigskip

\paragraph{Companion document.} The full mathematical derivation underlying the DDEVD --- including detailed proofs of the asymptotic MISE expansion, the optimal bandwidth theorem, the stability criterion, the closed-form coefficient formulae, and accompanying numerical experiments --- is provided as a self-contained companion technical report~\cite{sandbichler2026ddevd_companion} (\texttt{arXiv:2605.21416}). This Supplementary Information sketches the same material in a notation aligned with the main text; theorem labels match the companion report wherever practical, and we refer the reader to it for any proof details omitted here.
\bigskip

\section{The DDEVD Estimator: Definition and Target}\label{sec:SI_definition}

\subsection{Problem setting}

Let $X_{ij}$, $i = 1,\ldots,m$, $j = 1,\ldots,n_i$, be i.i.d.\ samples from a base distribution with CDF $F_X$.
The block structure models the natural partition of observations (e.g., years of daily precipitation records, or separate micrograph fields of view).
The quantity of interest is the \emph{metastatistical extreme value distribution}
\begin{equation}\label{eq:mev_target}
  F(y) = \frac{1}{m}\sum_{i=1}^{m} F_X^{n_i}(y),
\end{equation}
which is the average CDF of block maxima when the base distribution is exactly $F_X$.
The DDEVD estimates $F$ non-parametrically by replacing $F_X$ with a kernel density estimate:
\begin{equation}\label{eq:ddevd}
  \hat{F}_{h}(y) = \frac{1}{m}\sum_{i=1}^{m}
    \left(\frac{1}{n_i}\sum_{j=1}^{n_i} K\!\left(\frac{y - X_{ij}}{h_i}\right)\right)^{n_i},
\end{equation}
where $K(u) = \int_{-\infty}^{u} k(t)\,\dd t$ is a CDF kernel derived from a symmetric, zero-mean probability density $k$, and $h = (h_1,\ldots,h_m)^T$ is the bandwidth vector (one bandwidth per block).

\subsection{Connection to the MEV framework}

The target \eqref{eq:mev_target} is the Metastatistical Extreme Value (MEV) distribution introduced by Marani \& Ignaccolo~\cite{marani2015metastatistical}.
Existing MEV implementations substitute a \emph{parametric} estimate for $F_X$ (typically Weibull).
The DDEVD is the non-parametric counterpart: it substitutes the kernel CDF estimator, removing the parametric constraint at the cost of requiring a bandwidth selection rule.

\section{Asymptotic MISE Expansion and Optimal Bandwidth}\label{sec:SI_mise}

\subsection{Bias--variance decomposition}

The MISE of the DDEVD is
\[
  \MISE(\hat{F}_h) = \E\int \bigl(\hat{F}_h(y) - F(y)\bigr)^2\,\dd y
  = \int \Bigl[\Bias^2(\hat{F}_h)(y) + \Var(\hat{F}_h)(y)\Bigr]\,\dd y.
\]
We define per-block quantities $\Bias_i(y) = \E[F_{i,h}^{n_i}(y)] - F_X^{n_i}(y)$ and
$\Var_i(y) = \Var(F_{i,h}^{n_i})(y)$, so that
\[
  \Bias(\hat{F}_h)(y) = \frac{1}{m}\sum_i \Bias_i(y), \qquad
  \Var(\hat{F}_h)(y) = \frac{1}{m^2}\sum_i \Var_i(y).
\]

\subsection{Asymptotic expansion}

\begin{theorem}[Asymptotic MSE expansion]\label{thm:SI_expansion}
Under the assumptions that $F_X$ is twice continuously differentiable and $K$ is a CDF kernel with a zero-mean density ($\mu_{K,1} = 0$) and finite second moments for $K$ and $K^2$, the MSE at point $y$ admits the expansion
\begin{align*}
  \MSE(y) &= \frac{1}{m^2}\Bigl[\Bigl(\sum_i b_{0,i}\Bigr)^2 + \sum_i V_{0,i}\Bigr]\\
  &\quad + \frac{1}{m^2}\Bigl[2\Bigl(\sum_i b_{0,i}\Bigr)\Bigl(\sum_i \bar h_i\, b_{1,i}\Bigr) + \sum_i \bar h_i\, V_{1,i}\Bigr]\\
  &\quad + \frac{1}{m^2}\Bigl[\Bigl(\sum_i \bar h_i\, b_{1,i}\Bigr)^2
     + 2\Bigl(\sum_i b_{0,i}\Bigr)\Bigl(\sum_i \overline{h_i^2}\, b_{2a,i} + \bar h_i^2\, b_{2b,i}\Bigr)
     + \sum_i\bigl(\overline{h_i^2}\, V_{2a,i} + \bar h_i^2\, V_{2b,i}\bigr)\Bigr]\\
  &\quad + \mathcal{O}(h^3),
\end{align*}
where $\bar h_i = n_i^{-1}\sum_j h_{ij}$, $\overline{h_i^2} = n_i^{-1}\sum_j h_{ij}^2$,
and the coefficients $b_{s,i}(y)$, $V_{s,i}(y)$ are derived in Supplementary Section~\ref{sec:SI_coefficients}.
\end{theorem}

The proof proceeds by a binomial expansion of $F_{i,h}^{n_i}$ around $\E[F_{i,h}]$, followed by approximation of the resulting central moments via the CLT.

\subsection{Optimal bandwidth}

\begin{theorem}[Optimal bandwidth vector]\label{thm:SI_bandwidth}
The bandwidth vector $\mathbf{h}_{opt}$ minimising the second-order approximation of the MISE is
\begin{equation}\label{eq:SI_hopt}
  \mathbf{h}_{opt} = -\tfrac{1}{2}\,\mathbf{Q}^{-1}\mathbf{c},
\end{equation}
where $\mathbf{c} \in \R^m$ and $\mathbf{Q} \in \R^{m\times m}$ have components
\begin{align}
  c_i &= \int \Bigl(2\sum_j b_{0,j}(y)\cdot b_{1,i}(y) + V_{1,i}(y)\Bigr)\,\dd y,\\
  Q_{kl} &= \int \Bigl(b_{1,k}(y)\,b_{1,l}(y)
    + \delta_{kl}\bigl(2\sum_j b_{0,j}(y)\cdot b_{2,k}(y) + V_{2,k}(y)\bigr)\Bigr)\,\dd y,
\end{align}
with $b_{2,k} = b_{2a,k} + b_{2b,k}$ and $V_{2,k} = V_{2a,k} + V_{2b,k}$.
Moreover, $\mathbf{h}_{opt}$ is block-constant, i.e., $h_{ij} = h_i$ for all $j$ within block~$i$.
\end{theorem}

The optimality proof uses the stationarity condition $\partial\,\MISE/\partial h_{ij} = 0$, which shows that all per-observation bandwidths within a block must be equal, and then reads off the resulting linear system.

\subsection{Iterative plug-in algorithm}\label{sec:SI_methods_iterative}

Since $\mathbf{Q}$ and $\mathbf{c}$ depend on the unknown $F_X$ and $f_X$, we estimate them from the data via a fixed-point iteration initialised with Silverman's rule~\cite{silverman2018density}:
\[
  h_i^{(0)} = 1.06\,\hat\sigma_i\, n_i^{-1/5}.
\]
At each iteration, pilot density and CDF estimates are substituted into the integrands, and the raw update $\mathbf{h}^* = -\tfrac{1}{2}(\hat{\mathbf{Q}}^{(k)})^{-1}\hat{\mathbf{c}}^{(k)}$ is applied with damping factor $\lambda = 0.5$ to stabilise convergence.
The iteration terminates when the relative bandwidth change drops below $10^{-4}$.

\section{Stability of the Bandwidth Optimisation}\label{sec:SI_stability}

\subsection{Positive definiteness criterion}

For equal block sizes ($n_i = n$) and a zero-mean kernel, the Hessian simplifies to
\[
  \mathbf{Q} = b\,\mathbf{E} + a\,\mathbf{I},
\]
where $\mathbf{E}$ is the all-ones matrix, $b = \int b_1(y)^2\,\dd y > 0$, and
$a \propto D := \int \bigl(2m\,b_0(y)\,b_2(y) + V_2(y)\bigr)\,\dd y$.
The eigenvalues of $\mathbf{Q}$ are $a$ (multiplicity $m-1$) and $a + mb$ (multiplicity $1$), so
\begin{lemma}
  $\mathbf{Q}$ is positive definite if and only if $D > 0$.
\end{lemma}

\subsection{Asymptotic analysis of $D$}

For large $n$, the integrals in $D$ are evaluated by a change of variables $t = 1 - F_X(y)$ and Watson's lemma.
For distributions in the Fr\'{e}chet, Gumbel, or short-tailed Weibull class, the density satisfies $f_X \asymp t^{\gamma+1}$ near $t = 0$ (where $z = nt$), leading to
\[
  \int b_0\,b_2\,\dd y \asymp -C_b\,n^{-\gamma}, \qquad
  \int V_2\,\dd y \asymp C_V\, n^{1-\gamma/2},
\]
with $C_b, C_V > 0$.

\begin{theorem}[Asymptotic stability criterion]\label{thm:SI_stability}
Let $n_i = n$ for all $i$, and suppose $F_X$ has extreme value index $\gamma > -\tfrac{1}{2}$.
Then $D > 0$ for large $n$ if and only if
\begin{equation}\label{eq:stability}
  m < C(\gamma, F_X, K)\cdot n^{1 + \gamma/2},
\end{equation}
where $C$ is a positive constant depending on $\gamma$, $F_X$, and the kernel $K$.
\end{theorem}

\begin{proof}
$D > 0$ requires $2m\int b_0 b_2 + \int V_2 > 0$, i.e.\
$m < -(\int V_2)/(2\int b_0 b_2)$.
Substituting the asymptotic scalings:
\[
  m < \frac{C_V\, n^{1-\gamma/2}}{2\,C_b\,n^{-\gamma}} = \frac{C_V}{2C_b}\,n^{1+\gamma/2}.
\]
\end{proof}

For unequal block sizes, a perturbation argument (Weyl's inequality) suggests that the stability bound~\eqref{eq:stability} holds approximately with $n$ replaced by the mean block size $\bar n$, provided the relative standard deviation of block sizes is small; this is stated as Conjecture~2.1 (with proof sketch) in the companion report~\cite{sandbichler2026ddevd_companion}.

\section{Coefficient Formulae}\label{sec:SI_coefficients}

The asymptotic forms of the bias and variance coefficients for a zero-mean kernel are listed below.
Let $c = n^{-1/2}\sqrt{(1-F_X)/F_X}$,
$\rho(N,c) = (Nc^2+1)^{-1/2}\exp\{N^2c^2/[2(Nc^2+1)]\}$,
$z_1(N,c) = c(N-Nc^2-1)(Nc^2+1)^{-2}$, and
$\alpha(y) = f_X(y)\,\mu_{K^2,1}$,\;
$\beta(y) = \tfrac{1}{2}f'_X(y)\mu_{K^2,2} - f'_X(y)F_X(y)\mu_{K,2}$,
where $\mu_{K,p} = \int u^p k(u)\,\dd u$ and $\mu_{K^2,p} = \int u^p [K^2]'(u)\,\dd u$.

\medskip\noindent\textbf{Bias coefficients} ($\mu_{K,1}=0$):
\begin{align*}
  b_{0,i}(y) &\approx F_X^{n_i}(y)\Bigl(\sqrt{F_X(y)}\,\ee^{n_i(1-F_X(y))/2} - 1\Bigr),\\
  b_{1,i}(y) &\approx -\alpha(y)\,\frac{\sqrt{n_i}}{2\sqrt{1-F_X(y)}}\,F_X^{n_i}(y)
    \Bigl(\sqrt{n_i F_X(1-F_X)} - 1\Bigr)\ee^{n_i(1-F_X)/2},\\
  b_{2a,i}(y) &\approx \beta(y)\,\frac{\sqrt{n_i}}{2\sqrt{1-F_X}}\,F_X^{n_i}
    \Bigl(\sqrt{n_iF_X(1-F_X)}-1\Bigr)\ee^{n_i(1-F_X)/2}\\
    &\quad + \tfrac{n_i}{2}f'_X(y)\mu_{K,2}\,F_X^{n_i-1}(y)\sqrt{F_X(y)}\,\ee^{(n_i-1)(1-F_X)/2}.
\end{align*}

\medskip\noindent\textbf{Variance coefficients} are expressed in terms of the same auxiliary functions evaluated at $N = 2n_i$ vs.\ $N = n_i$; their explicit forms are given in the companion technical report.

\section{Implementation and Hyperparameters}\label{sec:SI_methods}

The DDEVD estimator and all baselines were implemented in Python 3.13 using the open-source \texttt{ddevd} package (repository URL in Section~\ref{sec:SI_code}).
Numerical integrals use \texttt{scipy.integrate.quad} (Fortran QUADPACK) with default tolerances; bandwidth root-finding uses \texttt{numpy.linalg.solve}.
Reference baselines: GEV maximum-likelihood fits via \texttt{scipy.stats.genextreme.fit}, log-normal fits via \texttt{scipy.stats.lognorm.fit} with $\mathrm{loc}=0$ fixed, and MEV with Weibull base via the \texttt{mevpy} package (Marani-group reference implementation).
Below, parameter names in \texttt{teletype} font correspond verbatim to keyword arguments of \texttt{ddevd.DDEVD}.

\subsection{Rainfall application}\label{sec:SI_methods_rainfall}

\textbf{Dataset.}
Geosphere Austria TAWES (\emph{Teilautomatische Wetterstationen}) archive, comprising $n=378$ stations with at least 5 complete years of observations between 1984 and 2024.
For each station we use two variables: the daily total precipitation (\texttt{rain\_day}, mm) and the within-day maximum 5-minute precipitation intensity (\texttt{rain\_5min}, mm).
Years with fewer than $5$ wet days (intensity~$>0$) are excluded as incomplete; the wet-day threshold for the 5-min variable is the precision floor of the rain gauge (typically $0.1\,\mathrm{mm}$).
The annual block is the calendar year, and the block size $n_i$ is the number of wet days that year ($\bar n \approx 150$ for daily totals, $\bar n \approx 70$ for 5-min intensities).
The number of blocks $m$ used per fit is the hindcast window length, $m = 10$~years.

\textbf{DDEVD configuration.}
\texttt{h\_opt\_position="quantile\_0.95"} (q-MISE optimised in the upper $5\,\%$ tail), \texttt{transform="log"} (multi-decade dynamic range; the bandwidth in $y$-space is uninformative for short rainfall records), \texttt{target\_distribution=None} (empirical pilot estimator; no parametric prior on $F_X$).
The iterative plug-in is run to convergence with absolute tolerance $10^{-4}$ on the mean bandwidth.

\textbf{Hindcast protocol.}
For every station and variable we slide a window of $W = 10$ consecutive years over the complete record (step $1\,$yr) and fit each estimator independently on each window.
For return periods $T \in \{2, 5, 10, 20, 30, 50, 100\}\,\mathrm{yr}$, we report the per-window predicted return level $\widehat{\mathrm{RL}}_T$.
The reference truth $\mathrm{RL}_{\mathrm{truth}}$ at each station is the Cunnane plotting position at $T$ computed from the full-record annual maxima.
The calibration ratio $\widehat{\mathrm{RL}}_T / \mathrm{RL}_{\mathrm{truth}}$ is pooled over all (station, window) pairs in Fig.~\ref{fig:rainfall}c,e.

\textbf{Competing methods.}
GEV is fit by maximum likelihood on the $W$ annual maxima of each window using \texttt{scipy.stats.genextreme.fit}; we cap MLE blow-ups by excluding windows where $\widehat{\mathrm{RL}}_T > 10 \times \max_i \mathrm{RL}_T^{\mathrm{emp}}$.
MEV is fit with a Weibull base distribution on the same wet-day intensities; the per-year shape and scale parameters are estimated by L-moments and the metastatistical CDF is averaged over the $W$ years.

\subsection{Metallurgy application}\label{sec:SI_methods_metallurgy}

\textbf{Dataset.}
$58$ optical micrographs of a polycrystalline alloy, $1\,300 \times 1\,080$ pixels each at a sample-plane pixel size of $2.5\,\mu\mathrm{m} / 6\,\mathrm{x} = 0.4167\,\mu\mathrm{m}$ per pixel (camera $2.5\,\mu\mathrm{m}$ pitch, optical magnification $6\,$x), corresponding to $0.244\,\mathrm{mm}^{2}$ of material per tile and $14.14\,\mathrm{mm}^{2}$ in aggregate.
Each tile was segmented with Segment-Anything-Model (SAM) \texttt{vit-l} (Meta~AI)~\cite{kirillov2023segment} with default mask-generation parameters; the output is filtered to retain masks with $100 \leq \mathrm{area} \leq 10^{5}\,\mathrm{px}^{2}$ (removing image-noise specks below the resolution limit, and the matrix/background super-segment that always exceeds $10^{5}\,\mathrm{px}^{2}$).
After the area filter and removal of $1\,060$ masks flagged as grain-merge/twin segmentation artefacts by a documented visual audit (\texttt{grain\_blocklist.py}; the largest oversize masks, each containing resolvable sub-grains), the total grain count is $n_{\mathrm{tot}} = 104\,321$ and the per-tile mean is $\bar{n} = 1\,799$ grains.

\textbf{DDEVD configuration.}
\texttt{h\_opt\_position="quantile\_0.9"} (q-MISE optimised in the upper $10\,\%$ tail), \texttt{transform="log"} (grain areas span $\sim 3$ decades), \texttt{target\_distribution=stats.norm} (Gaussian pilot on $\ln(\mathrm{area})$; the distribution-diagnostic figure in the supplementary figure package shows $\ln(\mathrm{area})$ is approximately Gaussian on the bulk).
The bandwidth optimisation uses \texttt{max\_bins=10} to bound the Q-matrix construction cost; F$_X$ averaging uses all $58$ tiles regardless.
The DDEVD uncertainty band in Fig.~\ref{fig:metal} pools $20$ leave-one-tile-out jackknife and $30$ half-data bootstrap resamples (the latter drawing $\lceil m/2 \rceil$ tiles without replacement and applied identically to the log-normal and GEV baselines); the reported curve is the median fit and the shaded band the $25\text{--}75\,\%$ resample range.

\textbf{Sampling-area scaling.}
The predicted median diameter of the largest grain in a sample of area $A$ is computed by
\begin{equation}
  d_{\max}(A) = 2 \sqrt{\widehat F_X^{-1}\!\bigl(0.5^{1/N(A)}\bigr) \cdot s_{\mathrm{px}}^2 / \pi},
  \quad N(A) = (A / A_{\mathrm{tile}}) \cdot \bar n,
\end{equation}
where $s_{\mathrm{px}} = 0.4167\,\mu\mathrm{m}$ is the sample-plane pixel size and $A_{\mathrm{tile}} = 0.244\,\mathrm{mm}^{2}$.

\textbf{Competing methods.}
The log-normal MLE in Fig.~\ref{fig:metal} is fit to the pooled grain areas (no block structure) using \texttt{scipy.stats.lognorm.fit} with $\mathrm{loc}=0$; the prediction at area $A$ uses the same $N(A)$ formula.
The GEV MLE is fit to the $58$ tile-level maxima $\ln(\mathrm{area}_{\max,i})$ using \texttt{scipy.stats.genextreme.fit}; the prediction at area $A$ uses the standard GEV $k$-block scaling with $k = A / A_{\mathrm{tile}}$.

\subsection{Synthetic stability validation}\label{sec:SI_methods_stability}

\textbf{Reference distributions.}
Normal($0,1$) and Exponential($1$) (Gumbel domain, $\gamma = 0$); Pareto with shape $\alpha = 2$ (Fréchet, $\gamma = 0.5$); Cauchy (Fréchet, $\gamma = 1$).

\textbf{Phase-diagram grid.}
$\log$-spaced $14 \times 14$ grid of $(n, m)$ with $n \in [150, 15\,000]$ and $m \in [10, 10^{5}]$; at each cell we draw one synthetic sample, compute the plug-in Hessian, and record $\mathrm{sign}(D_q)$ at $q = 0.95$.

\textbf{Empirical boundary.}
For each grid value of $n$ we binary-search on $m$ until the first $m$ with $D_q \leq 0$; the geometric mean of that $m$ and the immediately-preceding stable $m$ is the empirical boundary point.
The theoretical line $m = C \cdot n^{1+\gamma/2}$ is overlaid by least-squares fit of $\log m^\star$ vs.\ $\log n$ over the empirical boundary points (one fit per distribution).
A coarser $7 \times 9$ grid is available via \texttt{--fast} for development iteration.

\section{Cross-reference index}\label{sec:SI_xref}

The following table maps SI claims to their proofs in the companion technical report~\cite{sandbichler2026ddevd_companion} (henceforth ``\textbf{TR}''; available at \texttt{arXiv:2605.21416}).
The TR uses long-form theorem statements with full proofs; the SI uses condensed statements aligned with the main-paper notation.
Section labels in both documents are preserved across revisions for stable citation.

\begin{center}
\begin{tabular}{l l l}
\toprule
\textbf{SI label} & \textbf{Statement} & \textbf{Location in TR} \\
\midrule
\ref{sec:SI_definition}                  & DDEVD definition, MEV connection    & TR Sec. 1.1 \\
\ref{sec:SI_mise}, Thm~\ref{thm:SI_expansion}  & Asymptotic MISE expansion           & TR Thm. 2.1; proof App~A \\
\ref{sec:SI_mise}, Eq.~\eqref{eq:SI_hopt}       & Optimal bandwidth                   & TR Thm. 2.2; proof App~B \\
\ref{sec:SI_methods_iterative}            & Iterative plug-in                   & TR Sec. 2.1 \\
\ref{sec:SI_stability}                   & Stability criterion $D > 0$     & TR Sec. 2.2 \\
\ref{sec:SI_stability}, Thm~\ref{thm:SI_stability} & $m < C\,n^{1+\gamma/2}$ scaling     & TR Thm. 2.3; proof App~C \\
\ref{sec:SI_coefficients}                & Bias / variance coefficients        & TR App~A.2--A.3 \\
Fig.~\ref{fig:stability}                   & Phase-diagram numerical validation  & TR Secs. 4.1.1/2 \\
\bottomrule
\end{tabular}
\end{center}

\section{Data and Code Availability}\label{sec:SI_data}

\textbf{Rainfall data.}
The Geosphere Austria TAWES precipitation records used in the rainfall application are openly available at the Geosphere data hub: \\
\url{https://data.hub.geosphere.at/dataset/tawes-v1-10min} (10-min resolution; daily aggregates derived therein).
The 5-min daily maxima are available on request from Geosphere Austria for the same station list.
A snapshot of the cleaned, year-segmented blocks (one CSV per station) used in this study will be deposited at Zenodo (\texttt{[DOI to be assigned]}) at the time of acceptance to make the precise data subset citable.

\textbf{Micrograph data.}
The $58$ SAM-segmented optical-microscopy tiles used in the metallurgy application, comprising the cleaned segmentations with micrograph boundaries removed, are openly available at Zenodo (\url{https://doi.org/10.5281/zenodo.20621678}, v1.0.0) under the Apache-2.0 licence.
The raw, unsegmented optical-microscopy images can be requested from the corresponding author for archival purposes.

\textbf{Code.}\label{sec:SI_code}
The open-source \texttt{ddevd} Python package implementing the DDEVD estimator, all baselines (GEV, MEV, log-normal) used in this manuscript, and the figure-generation scripts is available at \\
\url{https://github.com/DataLabHell/ddevd} \\
under the MIT license.
The exact commit hash used to produce the figures in this submission, together with the generated data caches (parquet/CSV) and figure outputs (PDF/PNG), will be archived at Zenodo at the time of acceptance.
Reproducing the figures requires \texttt{uv} (or any standard Python virtual environment), pulling the dataset deposits, and running the per-figure scripts documented in the repository's \texttt{README}.


\bibliography{bibliography}

\end{document}